\documentclass[conference, letterpaper]{IEEEtran}
\usepackage{fancyhdr}
\usepackage{epsfig}
\usepackage{threeparttable}
\usepackage{epsf,epsfig}
\usepackage{amsthm}
\usepackage[cmex10]{amsmath}
\usepackage{amssymb, amsfonts}
\usepackage[noadjust]{cite}
\usepackage{dsfont}
\usepackage{subfigure}
\usepackage{color}
\usepackage{soul}
\usepackage{amssymb}
 \usepackage{accents}
  \usepackage{algorithm}
   \usepackage[english]{babel}
    \usepackage{url}
    \usepackage{hyperref}
 \usepackage[top=0.7in, bottom=1in,left=0.65in, right=0.65in]{geometry}

\begin{document}
\newtheorem{theorem}{Theorem}
\newtheorem{acknowledgement}[theorem]{Acknowledgement}
\newtheorem{axiom}[theorem]{Axiom}
\newtheorem{case}[theorem]{Case}
\newtheorem{claim}[theorem]{Claim}
\newtheorem{conclusion}[theorem]{Conclusion}
\newtheorem{condition}[theorem]{Condition}
\newtheorem{conjecture}[theorem]{Conjecture}
\newtheorem{criterion}[theorem]{Criterion}
\newtheorem{definition}{Definition}
\newtheorem{exercise}[theorem]{Exercise}
\newtheorem{lemma}{Lemma}
\newtheorem{corollary}{Corollary}
\newtheorem{notation}[theorem]{Notation}
\newtheorem{problem}[theorem]{Problem}
\newtheorem{proposition}{Proposition}
\newtheorem{solution}[theorem]{Solution}
\newtheorem{summary}[theorem]{Summary}
\newtheorem{assumption}{Assumption}
\newtheorem{example}{\bf Example}
\newtheorem{remark}{\bf Remark}

\def\qed{$\Box$}
\def\QED{\mbox{\phantom{m}}\nolinebreak\hfill$\,\Box$}
\def\proof{\noindent{\emph{Proof:} }}
\def\poof{\noindent{\emph{Sketch of Proof:} }}
\def
\endproof{\hspace*{\fill}~\qed
\par
\endtrivlist\unskip}
\def\endproof{\hspace*{\fill}~\qed\par\endtrivlist\vskip3pt}

\def\E{\mathsf{E}}
\def\eps{\varepsilon}
\def\phi{\varphi}
\def\Lsp{{\boldsymbol L}}
\def\Bsp{{\boldsymbol B}}
\def\lsp{{\boldsymbol\ell}}
\def\Ltsp{{\Lsp^2}}
\def\Lpsp{{\Lsp^p}}
\def\Linsp{{\Lsp^{\infty}}}
\def\LtR{{\Lsp^2(\Rst)}}
\def\ltZ{{\lsp^2(\Zst)}}
\def\ltsp{{\lsp^2}}
\def\ltZt{{\lsp^2(\Zst^{2})}}
\def\ninN{{n{\in}\Nst}}
\def\oh{{\frac{1}{2}}}
\def\grass{{\cal G}}
\def\ord{{\cal O}}
\def\dist{{d_G}}
\def\conj#1{{\overline#1}}
\def\ntoinf{{n \rightarrow \infty }}
\def\toinf{{\rightarrow \infty }}
\def\tozero{{\rightarrow 0 }}
\def\trace{{\operatorname{trace}}}
\def\ord{{\cal O}}
\def\UU{{\cal U}}
\def\rank{{\operatorname{rank}}}
\def\acos{{\operatorname{acos}}}

\def\SINR{\mathsf{SINR}}
\def\SNR{\mathsf{SNR}}
\def\SIR{\mathsf{SIR}}
\def\tSIR{\widetilde{\mathsf{SIR}}}
\def\Ei{\mathsf{Ei}}
\def\l{\left}
\def\r{\right}
\def\({\left(}
\def\){\right)}
\def\lb{\left\{}
\def\rb{\right\}}

\setcounter{page}{1}

\newcommand{\eref}[1]{(\ref{#1})}
\newcommand{\fig}[1]{Fig.\ \ref{#1}}

\def\bydef{:=}
\def\ba{{\mathbf{a}}}
\def\bb{{\mathbf{b}}}
\def\bc{{\mathbf{c}}}
\def\bd{{\mathbf{d}}}
\def\bee{{\mathbf{e}}}
\def\bff{{\mathbf{f}}}
\def\bg{{\mathbf{g}}}
\def\bh{{\mathbf{h}}}
\def\bi{{\mathbf{i}}}
\def\bj{{\mathbf{j}}}
\def\bk{{\mathbf{k}}}
\def\bl{{\mathbf{l}}}
\def\bm{{\mathbf{m}}}
\def\bn{{\mathbf{n}}}
\def\bo{{\mathbf{o}}}
\def\bp{{\mathbf{p}}}
\def\bq{{\mathbf{q}}}
\def\br{{\mathbf{r}}}
\def\bs{{\mathbf{s}}}
\def\bt{{\mathbf{t}}}
\def\bu{{\mathbf{u}}}
\def\bv{{\mathbf{v}}}
\def\bw{{\mathbf{w}}}
\def\bx{{\mathbf{x}}}
\def\by{{\mathbf{y}}}
\def\bz{{\mathbf{z}}}
\def\b0{{\mathbf{0}}}

\def\bA{{\mathbf{A}}}
\def\bB{{\mathbf{B}}}
\def\bC{{\mathbf{C}}}
\def\bD{{\mathbf{D}}}
\def\bE{{\mathbf{E}}}
\def\bF{{\mathbf{F}}}
\def\bG{{\mathbf{G}}}
\def\bH{{\mathbf{H}}}
\def\bI{{\mathbf{I}}}
\def\bJ{{\mathbf{J}}}
\def\bK{{\mathbf{K}}}
\def\bL{{\mathbf{L}}}
\def\bM{{\mathbf{M}}}
\def\bN{{\mathbf{N}}}
\def\bO{{\mathbf{O}}}
\def\bP{{\mathbf{P}}}
\def\bQ{{\mathbf{Q}}}
\def\bR{{\mathbf{R}}}
\def\bS{{\mathbf{S}}}
\def\bT{{\mathbf{T}}}
\def\bU{{\mathbf{U}}}
\def\bV{{\mathbf{V}}}
\def\bW{{\mathbf{W}}}
\def\bX{{\mathbf{X}}}
\def\bY{{\mathbf{Y}}}
\def\bZ{{\mathbf{Z}}}

\def\mA{{\mathbb{A}}}
\def\mB{{\mathbb{B}}}
\def\mC{{\mathbb{C}}}
\def\mD{{\mathbb{D}}}
\def\mE{{\mathbb{E}}}
\def\mF{{\mathbb{F}}}
\def\mG{{\mathbb{G}}}
\def\mH{{\mathbb{H}}}
\def\mI{{\mathbb{I}}}
\def\mJ{{\mathbb{J}}}
\def\mK{{\mathbb{K}}}
\def\mL{{\mathbb{L}}}
\def\mM{{\mathbb{M}}}
\def\mN{{\mathbb{N}}}
\def\mO{{\mathbb{O}}}
\def\mP{{\mathbb{P}}}
\def\mQ{{\mathbb{Q}}}
\def\mR{{\mathbb{R}}}
\def\mS{{\mathbb{S}}}
\def\mT{{\mathbb{T}}}
\def\mU{{\mathbb{U}}}
\def\mV{{\mathbb{V}}}
\def\mW{{\mathbb{W}}}
\def\mX{{\mathbb{X}}}
\def\mY{{\mathbb{Y}}}
\def\mZ{{\mathbb{Z}}}

\def\cA{\mathcal{A}}
\def\cB{\mathcal{B}}
\def\cC{\mathcal{C}}
\def\cD{\mathcal{D}}
\def\cE{\mathcal{E}}
\def\cF{\mathcal{F}}
\def\cG{\mathcal{G}}
\def\cH{\mathcal{H}}
\def\cI{\mathcal{I}}
\def\cJ{\mathcal{J}}
\def\cK{\mathcal{K}}
\def\cL{\mathcal{L}}
\def\cM{\mathcal{M}}
\def\cN{\mathcal{N}}
\def\cO{\mathcal{O}}
\def\cP{\mathcal{P}}
\def\cQ{\mathcal{Q}}
\def\cR{\mathcal{R}}
\def\cS{\mathcal{S}}
\def\cT{\mathcal{T}}
\def\cU{\mathcal{U}}
\def\cV{\mathcal{V}}
\def\cW{\mathcal{W}}
\def\cX{\mathcal{X}}
\def\cY{\mathcal{Y}}
\def\cZ{\mathcal{Z}}
\def\cd{\mathcal{d}}
\def\Mt{M_{t}}
\def\Mr{M_{r}}
\def\O{\Omega_{M_{t}}}
\newcommand{\figref}[1]{{Fig.}~\ref{#1}}
\newcommand{\tabref}[1]{{Table}~\ref{#1}}

\newcommand{\var}{\mathsf{var}}
\newcommand{\fb}{\tx{fb}}
\newcommand{\nf}{\tx{nf}}
\newcommand{\BC}{\tx{(bc)}}
\newcommand{\MAC}{\tx{(mac)}}
\newcommand{\Pout}{p_{\mathsf{out}}}
\newcommand{\nnn}{\nn\\}
\newcommand{\FB}{\tx{FB}}
\newcommand{\TX}{\tx{TX}}
\newcommand{\RX}{\tx{RX}}
\renewcommand{\mod}{\tx{mod}}
\newcommand{\m}[1]{\mathbf{#1}}
\newcommand{\td}[1]{\tilde{#1}}
\newcommand{\sbf}[1]{\scriptsize{\textbf{#1}}}
\newcommand{\stxt}[1]{\scriptsize{\textrm{#1}}}
\newcommand{\suml}[2]{\sum\limits_{#1}^{#2}}
\newcommand{\sumlk}{\sum\limits_{k=0}^{K-1}}
\newcommand{\eqhsp}{\hspace{10 pt}}
\newcommand{\tx}[1]{\texttt{#1}}
\newcommand{\Hz}{\ \tx{Hz}}
\newcommand{\sinc}{\tx{sinc}}
\newcommand{\tr}{\mathrm{tr}}
\newcommand{\diag}{\mathrm{diag}}
\newcommand{\MAI}{\tx{MAI}}
\newcommand{\ISI}{\tx{ISI}}
\newcommand{\IBI}{\tx{IBI}}
\newcommand{\CN}{\tx{CN}}
\newcommand{\CP}{\tx{CP}}
\newcommand{\ZP}{\tx{ZP}}
\newcommand{\ZF}{\tx{ZF}}
\newcommand{\SP}{\tx{SP}}
\newcommand{\MMSE}{\tx{MMSE}}
\newcommand{\MINF}{\tx{MINF}}
\newcommand{\RC}{\tx{MP}}
\newcommand{\MBER}{\tx{MBER}}
\newcommand{\MSNR}{\tx{MSNR}}
\newcommand{\MCAP}{\tx{MCAP}}
\newcommand{\vol}{\tx{vol}}
\newcommand{\ah}{\hat{g}}
\newcommand{\tg}{\tilde{g}}
\newcommand{\teta}{\tilde{\eta}}
\newcommand{\heta}{\hat{\eta}}
\newcommand{\uh}{\m{\hat{s}}}
\newcommand{\eh}{\m{\hat{\eta}}}
\newcommand{\hv}{\m{h}}
\newcommand{\hh}{\m{\hat{h}}}
\newcommand{\Po}{P_{\mathrm{out}}}
\newcommand{\Poh}{\hat{P}_{\mathrm{out}}}
\newcommand{\Ph}{\hat{\gamma}}
\newcommand{\mat}[1]{\begin{matrix}#1\end{matrix}}
\newcommand{\ud}{^{\dagger}}
\newcommand{\C}{\mathcal{C}}
\newcommand{\nn}{\nonumber}
\newcommand{\nInf}{U\rightarrow \infty}

\addtolength{\textfloatsep}{-6mm}
\setlength{\abovecaptionskip}{-0.5pt}
\textheight=24.5cm





\setlength{\belowcaptionskip}{-15pt}

\title{\huge Multiuser Resource Allocation for \\ Mobile-Edge Computation Offloading} 
\author{\IEEEauthorblockN{Changsheng You and Kaibin Huang}
\IEEEauthorblockA{Department of Electrical and Electronic Engineering\\
The University of Hong Kong\\
Email: csyou@eee.hku.hk, huangkb@eee.hku.hk}}\vspace{-10pt}


\maketitle
\begin{abstract} 
\emph{Mobile-edge computation offloading} (MECO) offloads intensive mobile computation to clouds located at the edges of cellular networks. Thereby, MECO is envisioned as a  promising technique for  prolonging  the battery lives and enhancing the computation capacities of mobiles.  In this paper, we consider resource allocation in a  MECO system comprising multiple users that time share  a single edge cloud and have different computation loads. The optimal resource allocation is formulated as a convex optimization problem  for minimizing  the weighted sum   mobile energy consumption under constraint on computation latency  and for both the cases of infinite and finite edge cloud computation capacities. The optimal policy is proved to have a threshold-based structure with respect to a derived \emph{offloading priority function}, which  yields priorities  for users according  to their channel gains and local computing energy consumption.  As a result, users with priorities above and below  a given  threshold perform complete  and minimum offloading, respectively.  Computing the threshold requires iterative computation. To reduce the complexity, a sub-optimal resource-allocation algorithm  is proposed and shown by simulation to have close-to-optimal performance. 

\end{abstract}

\section{Introduction}

The realization of Internet of Things (IoT) will connect tens of billions of resource-limited  mobiles, e.g., mobile devices, sensors and wearable computing devices, to Internet via cellular networks. The finite battery lives and limited computation capacities of mobiles pose  challenges for designing  IoT. One promising solution is to leverage mobile-edge computing \cite{patel2014mobile} and offload intensive mobile computation to nearby clouds at the edges of cellular networks, called \emph{edge clouds},  with short latency, referred to as \emph{mobile-edge computation offloading } (MECO).  In this paper, we consider a MECO system with a single edge cloud serving multiple users and investigate the energy-efficient resource allocation.  

Mobile computation offloading (MCO) (or mobile cloud computing) has been extensively studied in  computer science, including system architectures \cite{Dinh:MCCSurver:2013}, virtual machine migration \cite{xiao2013dynamic} and server consolidation \cite{srikantaiah2008energy}.  It is commonly assumed that the implementation of MCO relies on a network architecture with a central cloud (e.g., a data center). This architecture has the drawbacks of high overhead  and long backhaul latency \cite{Ahmed2015ahmedsurvey} and will soon encounter the performance bottleneck of finite backhaul capacity in view of exponential mobile traffic  growth. These issues can be overcome by MECO based on a network architecture supporting distributed mobile-edge computing. 

Energy efficient MECO requires the joint design of MCO and wireless communication techniques. Recent years have seen research progress on this topic.  For a single-user MECO system, the optimal offloading decision policy was derived in \cite{zhang:MobileMmodel:2013} by comparing the energy consumption of optimized local computing (with variable CPU cycles) and offloading (with  variable transmission rates). This framework was  further developed in \cite{you:MCC:2015} and \cite{dynamic2016zhang} to enable adaptive offloading powered by wireless  energy transfer and energy harvesting, respectively. In \cite{xiang2014energy}, also for a single-user MECO system, dynamic offloading was integrated with adaptive  LTE/WiFi link selection. Moreover, resource allocation for MECO  has been studied for  various types of multiuser systems \cite{Barbarossa:MobileCloud:2014,zhao2015cooperative, xu2015mec}.  In \cite{Barbarossa:MobileCloud:2014}, considering a multi-cell MECO system, the radio and computation resources were jointly allocated to minimize the mobile energy consumption under offloading latency  constraints.  With the coexistence of central and edge clouds, the optimal user scheduling for offloading to different clouds was studied in \cite{zhao2015cooperative}.  In addition, the distributed  offloading for multiuser MECO was designed in \cite{xu2015mec} using  game theory for both energy-and-latency minimization.  Prior work on MECO resource allocation focuses on complex algorithmic designs and yields little insight into the optimal policy structures. In contrast, for a multiuser MECO system based on time-division multiple access (TDMA), the optimal resource-allocation policy is shown in current work to have a simple threshold-based structure with respect to a derived offloading priority function. 

Resource allocation has been widely studied   for  various types of  multiuser  communication systems, e.g., TDMA (see e.g., \cite{wang2008power}),  orthogonal frequency-division multiple access (OFDMA) (see e.g., \cite{wong1999multiuser}) and code-division multiple access (CDMA) (see e.g., \cite{oh2003optimal}).  Note that all of them only focus on the radio resource allocation. In contrast, for newly proposed MECO systems,  both the computation and radio resource allocation at edge clouds need to be jointly  optimized for the maximum mobile energy savings, which makes the algorithmic design more complex.

This paper considers a multiuser MECO system based on TDMA. Consider both the cases of infinite and finite cloud computation capacities. The optimal resource-allocation policy is derived by solving a convex optimization problem that minimizes the weighted sum mobile energy consumption. Note that the consideration of MECO simplifies the problem formulation since the long backhaul latency and heavy overhead in central clouds can be neglected.  To solve the problem, an \emph{offloading priority function} is derived  that yields priorities for users and depends on their channel gains and local computing energy consumption. Based on this, the optimal  policy is proved to have an insightful threshold-based structure that determines complete or minimum offloading for users with priorities above or below a given threshold, respectively. Moreover, to reduce the complexity for  computing the threshold, a simple sub-optimal resource-allocation algorithm is designed and shown to  have close-to-optimal performance by simulation.  

\section{System Model}\label{Sec:Sys}

Consider a multiuser MECO system shown in Fig.~\ref{Fig:Sys_Multiuser_MEC}(a) that comprises  $K$ single-antenna mobiles, indexed as $1, 2, \cdots, K$,  and one single-antenna base station (BS) that is the gateway  of an  edge cloud. Time is divided into slots each with a duration of $T$ seconds. As shown in Fig.~\ref{Fig:Sys_Multiuser_MEC}(a), each slot comprises two sequential phases for 1) mobile offloading or local computing and 2) cloud computing and downloading of computation results from the edge cloud to mobiles. Cloud computing has small latency; the downloading does not consume mobile energy and furthermore is much faster than offloading due to relative smaller sizes of computation results. For these reasons, the second phase is assumed to have a negligible duration compared with  the first phase  and not considered in resource allocation.   Considering an arbitrary  slot, the BS schedules a subset of users for complete/partial offloading based on TDMA. The user with partial or no offloading computes a fraction of or all input data, respectively, using a local CPU.  Moreover, the BS is assumed to have perfect knowledge of multiuser channel gains, local computing energy  per bit  and sizes of input data at all users. Using these information, the BS selects offloading users, determines the offloaded data sizes and allocates fractions of the slot to offloading users with the criterion of minimum weighted sum mobile energy consumption. In addition,  channels are assumed to remain  constant within  each slot.

The model of local computing is described as follows.  Assume that the CPU frequency is fixed at each user and may vary over users. Consider an arbitrary time slot. Following the model in \cite{xu2015mec}, let $C_k$ denote the number of CPU cycles required for computing $1$-bit of input data at the $k$-th mobile,  and $P_k$ the  energy consumption per  cycle for local computing at this user. Then the product $C_kP_k$ gives  computing energy per bit.  As shown in Fig.~\ref{Fig:Sys_Multiuser_MEC}(b), mobile  $k$ is required to compute $R_k$-bit input data within the slot, out of which $\ell_k$-bit is offloaded and $(R_k -\ell_k)$-bit is computed locally.  Then the total energy consumption for local computing at mobile  $k$, denoted as $E_{\text{loc},k}$, is given by $E_{\text{loc},k}= (R_k-\ell_k)C_k P_k$. Let $F_k$  denote the computation capacity of  mobile $k$ that is measured by the number of CPU cycles per second. Under the computation latency constraint, $C_k (R_k-\ell_k) \le F_k T$. As a result, the offloaded data at mobile $k$ has the minimum size of   $\ell_k \ge m_k^+$ with $m_k=R_k-\frac{F_k T}{C_k}$, where the function $(x)^+=\max\{x, 0\}.$

\begin{figure}[t]
\begin{center}
\subfigure[Multiuser MECO system.]{\includegraphics[width=7.5cm]{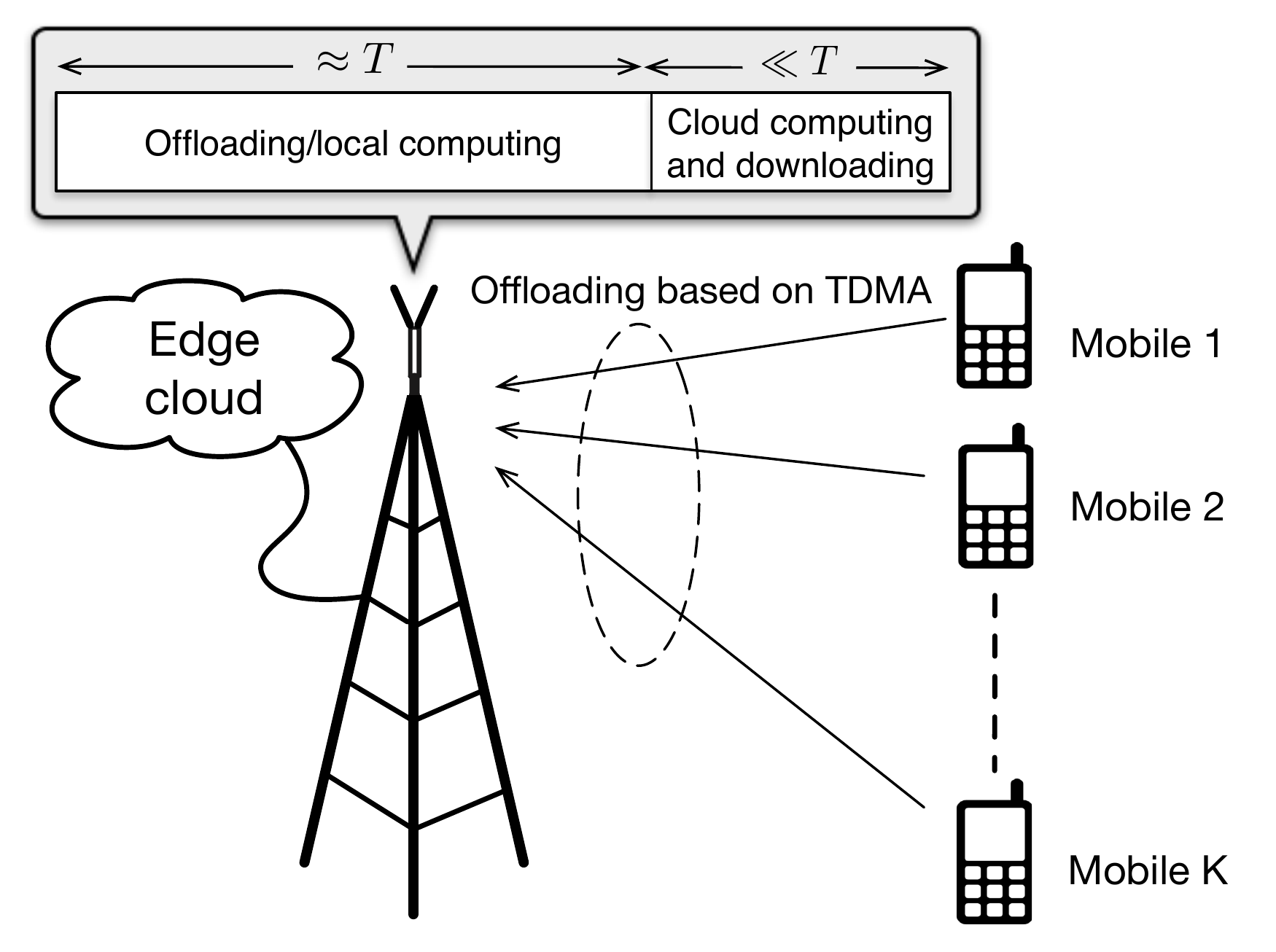}}\\
\vspace{-5pt}
\subfigure[Mobile computation offloading.]{\includegraphics[width=8.5cm]{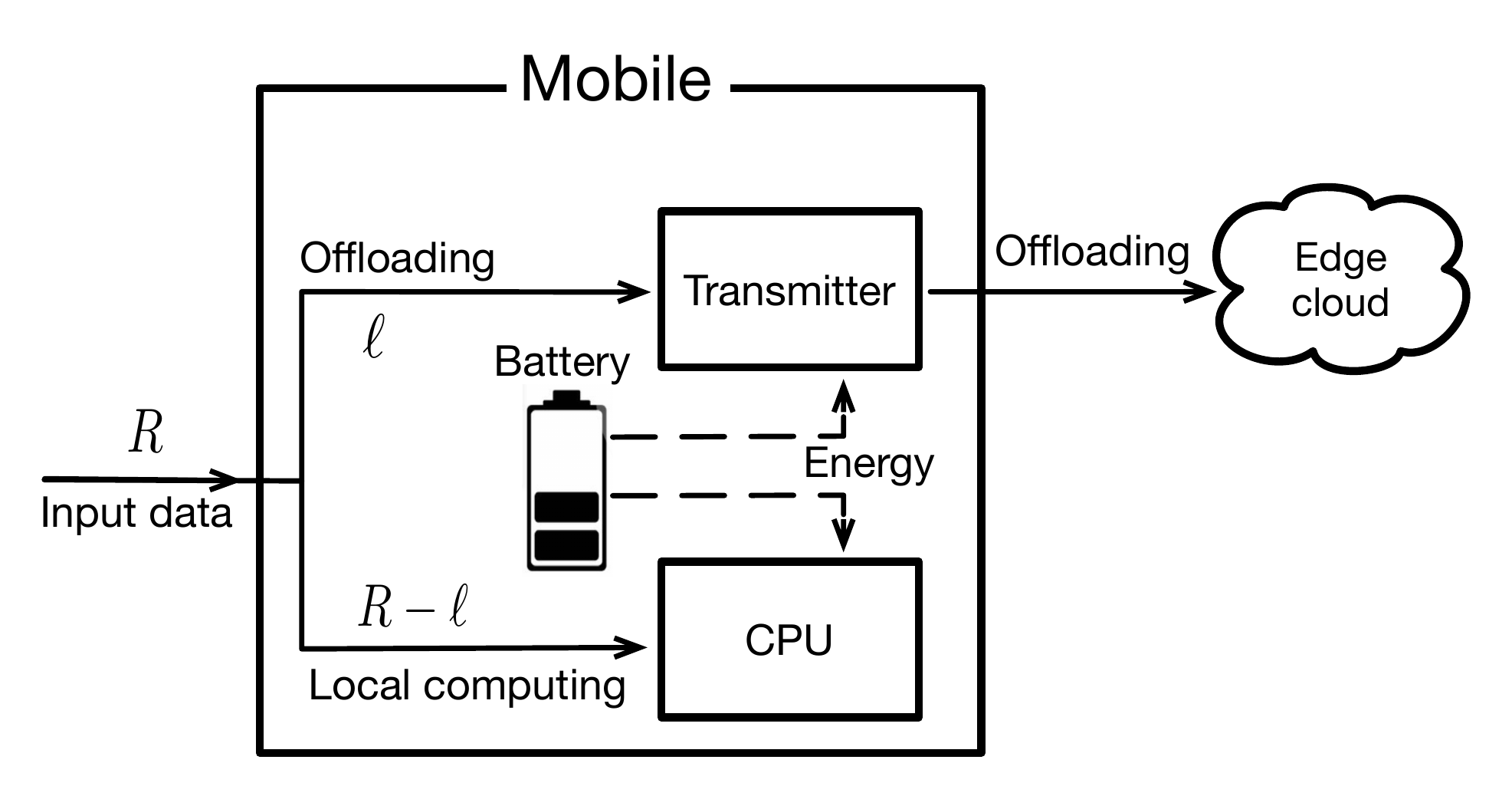}}
\caption{(a) Multiuser MECO system and (b) mobile computation offloading.}
\label{Fig:Sys_Multiuser_MEC}
\end{center}
\end{figure}

Next, the energy consumption for offloading  is modeled. Let $h_k$ denote the channel gain and $p_k$ the transmission power for mobile $k$.  Then the achievable rate, denoted by $r_k$, is given as:
\begin{equation}
r_k=B \log\l(1+\dfrac{p_k h_k^2}{N_0}\r) \label{Eq:TDMARate}
\end{equation} 
where  $N_0$ is the variance of complex white Gaussian channel noise. The fraction of slot allocated to mobile $k$ for offloading is denoted as $t_k$ with $t_k\geq 0$,   where $t_k =0$ corresponds to no offloading. For the case of offloading ($t_k > 0$), the transmission rate is fixed as $r_k\!\!\!=\!\!\!\ell_k/t_k$ since this is the most energy-efficient transmission policy under a deadline constraint. Define a function $f(x)=N_0 (2^{\frac{x}{B}}-1)$. It follows from \eqref{Eq:TDMARate} that the energy consumption  for offloading at  mobile $k$ is
\begin{equation}
E_{\text{off},k}=p_k t_k=\dfrac{t_k}{h_k^2} f\!\l(\dfrac{\ell_k}{t_k}\r). \label{Eq:OffEgy}
\end{equation} 
Note that if either $\ell_k=0$ or $t_k = 0$, $E_{\text{off},k}$ is equal to zero.

Last, consider the edge cloud. It  is assumed that the edge cloud has finite computation capacity, denoted as $F$, measured as the maximum CPU cycles allowed for computing the sum  offloaded data in each slot: $\sum_{k=1}^{K} C_k \ell_k \le F $. This constraint ensures low latency for cloud computing.

\section{Multiuser  MECO: Problem Formulation}

In this section, resource allocation for multiuser MECO is formulated as an optimization problem. The objective is to minimize the  weighted sum mobile energy consumption: $\sum_{k=1}^K \beta_k (E_{\text{off},k}+E_{\text{loc},k})$,  where the positive weight factors $\{\beta_k\}$ account for fairness among mobiles. Under the constraints on time-sharing, cloud computation capacity and computation latency, the resource allocation problem is formulated as follows: 
\begin{equation}\tag{$\textbf{P1}$} 
\begin{aligned}
\min_ {\{\ell_k, t_k\} }   ~ &\sum_{k=1}^{K} \beta_k \l[ \dfrac{t_k}{h_k^2}  f\!\l(\frac{\ell_k}{t_k}\r) + (R_k-\ell_k)C_k P_k \r] \\
\text{s.t.}\qquad 
& \sum_{k=1}^K t_k \le T, \quad \sum_{k=1}^{K} C_k \ell_k \le F,\\ 
& t_k \ge 0, \quad m_k^+\le \ell_k \le R_k, \quad \forall \ k. 
\end{aligned}
\end{equation}

Several basic characteristics of Problem P1 are given  in  the following two  lemmas.  

\begin{lemma}\label{Lem:ConvexProblem} \emph{Problem P1 is a convex optimization problem.}
\end{lemma}
\begin{proof}
See Appendix~\ref{App:ConvexProblem}.
\end{proof}

\begin{lemma}\label{Lem:FeasiTypeI} \emph{The feasibility condition for Problem P1 is: $\sum_{k=1}^K m_k^+ C_k \le F$.}
\end{lemma}
\begin{proof}
See Appendix~\ref{App:FeasiTypeI}.
\end{proof}

Lemma \ref{Lem:FeasiTypeI} shows that whether the cloud computation capacity constraint is satisfied  determines the feasibility of this optimization problem, while the time-sharing constraint can always be satisfied and only affects the mobile energy consumption.

Assume that Problem P1 is feasible. The direct solution of Problem P1 using the dual-decomposition approach (the  Lagrange method)  requires iterative computation and yields no insight into the structure of the optimal policy. To address these issues, we adopt a two-stage solution approach  that requires first solving Problem P2  below that relaxes  Problem P1 by removing the constraint on the cloud computation capacity: 
\begin{equation}\tag{$\textbf{P2}$} 
\begin{aligned}
\min_ {\{\ell_k, t_k\} }   ~ &\sum_{k=1}^{K} \beta_k \l[ \dfrac{t_k}{h_k^2}  f\!\l(\frac{\ell_k}{t_k}\r) + (R_k-\ell_k)C_k P_k \r] \\
\text{s.t.}\qquad 
& \sum_{k=1}^K t_k \le T,\\ 
& t_k \ge 0,  \quad m_k^+\le \ell_k \le R_k, \quad \forall \ k. 
\end{aligned}
\end{equation}
If the solution for Problem P2 violates the constraint on cloud computation capacity, Problem P1 is then incrementally solved building on the solution for Problem P2. This approach allows the optimal policy to be  shown to have the said  threshold-based structure and also facilitates the design of low-complexity close-to-optimal resource-allocation algorithm. It is interesting to note that Problem P2  corresponds to the case where the edge cloud has infinite computation capacity.  The detailed procedures for solving Problems P1 and P2 are presented in the subsequent two sections.

\section{Multiuser  MECO: Infinite Cloud Capacity}
In this section, by solving Problem P2 using the Lagrange method, we derive a threshold-based policy for the optimal resource allocation. Moreover, the policy is simplified for several special cases. 

To solve Problem P2, the  Lagrange function is defined as 
\begin{equation*}
L=\sum_{k=1}^{K} \beta_k \l[ \dfrac{t_k}{h_k^2}  f\l(\frac{\ell_k}{t_k}\r) + (R_k-\ell_k) C_k P_k \r]+ \lambda\l(\sum_{k=1}^K t_k - T\r)
\end{equation*}
where $\lambda \geq 0$ is the Lagrange multiplier associated with the time-sharing constraint. For ease of notation, define a function $g(x)=f(x)-x f^{'}(x)$. Let $\{\ell_k^{*(2)}, t_k^{*(2)}\}$ denote the solution for Problem P2 that always exists according to Lemma~\ref{Lem:FeasiTypeI}.  Then applying KKT conditions leads to the following necessary and sufficient conditions:
\vspace{-5pt}
\begin{subequations}
\begin{align}\label{Eq:OptL}
 & \dfrac{\partial L}{\partial \ell_k^{*(2)}}\!=\!\frac{\beta_k f^{'}\!\!\l(\frac{\ell_k^{*(2)}}{t_k^{*(2)}}\r)}{h_k^2}\!-\!\!\beta_k C_k P_k
\begin{cases}
>0, \!\!&\ell_k^{*(2)}\!=\!m_k^+\cr =0, \!&\ell_k^{*(2)}\in(m_k^+, R_k) \cr <0, \!&\ell_k^{*(2)}=R_k
\end{cases}, \forall k., \\ 
 & \dfrac{\partial L}{\partial t_k^{*(2)}}=\dfrac{ \beta_k g\!\l(\frac{\ell_k^{*(2)}}{t_k^{*(2)}}\r)}{h_k^2}+\lambda^*
\begin{cases}
  >0, & t_k^{*(2)}=0 \cr =0, &t_k^{*(2)}>0 
\end{cases}, 
\forall k.,  \label{Eq:OptT}\\
& \sum_{k=1}^K t_k^{*(2)}\le T,~~~~\lambda^* \l(\sum_{k=1}^K t_k^{*(2)} - T\r)=0. \label{Eq:OptT:a}
\end{align}
\end{subequations}

Based on these conditions,  the optimal policy for resource allocation is characterized in the following sub-sections. 
\subsection{Offloading Priority Function}
Define  an (mobile)  \emph{offloading priority function}, which is essential for the optimal resource allocation, as follows: 
\begin{equation}\label{Eq:OffPriority}
\!\!\!\phi(\beta_k, C_k,  P_k, h_k)\!=\!\!\begin{cases}
   \dfrac{\beta_k N_0}{h_k^2} \l( \upsilon_k  \ln \upsilon_k\!-\!\upsilon_k\!+\!1\r),  \!\!\!&\mbox{$\upsilon_k\ge 1$}\\
   0,  &\mbox{$\upsilon_k<1$}
   \end{cases},
\end{equation}
with the constant $\upsilon_k$ defined as 
\begin{equation}\label{Eq:Const}
\upsilon_k=\dfrac{B C_k  P_k h_k^2 }{N_0\ln2 }. 
\end{equation}

This function is derived by solving a useful equation as shown in the following lemma. 

\begin{lemma}\label{Lem:PriStrongLA}\emph{
Given $\upsilon_k\ge1$, the offloading priority function $\phi(\beta_k, C_k,  P_k, h_k)$ in \eqref{Eq:OffPriority} is   the root of the equation with respect to $x$: $$f^{'-1}\!\l(C_k P_k h_k^2\r)\!=\!g^{-1}\!\l(\frac{-h_k^2 x}{\beta_k} \r).$$ 
}
\end{lemma}
\begin{proof}
See Appendix~\ref{App:PriStrongLA}.
\end{proof}

The function  generates an offloading priority value, $\phi_k = \phi(\beta_k, C_k,  P_k, h_k)$,   for mobile $k$  depending on corresponding variables quantifying  fairness, local computing and channel. The amount of offloaded data by a mobile grows with an increasing offloading priority as shown in the next sub-section. It is useful to understand the effects of parameters on the offloading priority that are characterized as follows. 

\begin{lemma}\label{Lem:Phi}\emph{
Given $\upsilon \geq 1$, $\phi(\beta, C, P, h)$ is a \emph{monotone increasing function} for  $\beta$, $C$, $P$ and $h$.  
}
\end{lemma}

Lemma~\ref{Lem:Phi} can be easily proved by deriving the first derivatives of $\phi$ with respect to each parameter. Moreover, it is consistent  with the intuition that, to reduce  energy consumption  by offloading, the BS should schedule those  mobiles having  high computing energy  consumption per bit (i.e., large $C$ and $P$) or good channels (i.e., large $h$). 

\begin{remark}[Effects of parameters on the offloading priority]\emph{It can be observed from \eqref{Eq:OffPriority} and \eqref{Eq:Const} that the offloading priority scales with local computing energy per bit $CP$ approximately as $(CP)\ln(CP)$ and with the channel gain $h$ approximately as $\ln h$. The former scaling is much faster than the latter. This shows that the  computing energy per bit is dominant  over the  channel on determining whether to offload. }
\end{remark}

\subsection{Optimal Resource-Allocation Policy}

Based on  conditions in \eqref{Eq:OptL}-\eqref{Eq:OptT:a} and Lemma~\ref{Lem:PriStrongLA}, the main result of this section is derived,  given in the following theorem.
\begin{theorem}[Optimal Resource-Allocation Policy]\label{Theo:OptiPolicyP2}\emph{Consider the case of infinite cloud computation capacity. The optimal policy solving Problem P2 has the following structure. 
\begin{enumerate}
\item If $\upsilon_k \leq 1 $ and the minimum offloaded data size $m_k^+=0$ for all $k$,   none of these  users performs offloading, i.e., $$\ell_k^{*(2)}=t_k^{*(2)}=0 \quad \forall k.$$ 
\item If there exists mobile $k$ such that $\upsilon_k > 1$ or $m_k^+>0$, for $k = 1, 2, \cdots, K$, 
\begin{align*}
\ell_k^{*(2)}&
   \begin{cases}
   =m_k^+,   &\mbox{$\phi_k<\lambda^*$}\\
   \in [m_k^+, R_k],  &\mbox{$\phi_k=\lambda^*$}\\
   =R_k , &\mbox{$\phi_k>\lambda^*$}
   \end{cases},
   \end{align*}
and
 \begin{equation*}
t_k^{*(2)}=\frac{ \ln 2}{B\l[W_0\l(\frac{\lambda^{*} h_k^2/\beta_k-N_0}{N_0 e}\r)+1\r]} \times \ell_k^{*(2)}
\end{equation*}
where $W_0(x)$ is the Lambert function and $\lambda^*$ is the optimal value of the Lagrange multiplier. Furthermore, the time-sharing constraint is active: $\sum_{k=1}^{K} t_k^{*(2)} =  T$. 
\end{enumerate}
}
\end{theorem}
\begin{proof}
See Appendix~\ref{App:OptiPolicyP2}.
\end{proof}

Theorem~\ref{Theo:OptiPolicyP2} reveals that the optimal resource-allocation policy has a threshold-based structure when offloading saves energy. In other words,  since the exact case of $\phi_k=\lambda^*$ rarely occurs in practice, the optimal policy makes a \emph{binary offloading decision} for each mobile. Specifically, if the corresponding offloading priority exceeds a given threshold, the mobile should offload all input data to the edge cloud; otherwise,   the mobile should offload only the minimum amount of data under the computation latency constraint.  This result is consistent with the intuition that the  greedy method can lead to the optimal  resource allocation.

\begin{remark}[Offloading or not?]\emph{For a conventional TDMA communication  system, continuous transmission by at least one  mobile  is always advantageous under the criterion of minimum sum energy consumption \cite{wang2008power}. However, this does not always hold for a TDMA  MECO system where no offloading for all users may be preferred  as shown in Theorem~\ref{Theo:OptiPolicyP2}. There are two cases where offloading is necessary. First, there exists at least one   mobile whose input data size is too large such that complete local computing fails to meet the latency constraint. Second,  some  mobile has a sufficient high value for the product  $ C_k P_k h_k^2$, indicating  that energy savings can be achieved by offloading because of high channel gain or large local computing energy consumption. }
\end{remark}

\begin{remark}[Offloading rate]\emph{It can be observed from Theorem~\ref{Theo:OptiPolicyP2}  that the offloading rate, defined as $\ell_k^{*(2)}/t_k^{*(2)}$ for mobile $k$,  is determined only by the channel gain and fairness weight  factor while other factors, namely $C_k$ and  $P_k$, affect the offloading decision. The rate increases with a growing  channel gain and vice versa since a large channel gain supports a higher  transmission rate or reduces transmission power, making offloading desirable for reducing  mobile energy consumption. }
\end{remark}

\begin{remark}[Algorithm computation complexity]\emph{
The traditional method for solving Problem P2 is the block-coordinate descending which performs iterative optimization  of the two sets of variables,  $\{\ell_k\}$ and $\{t_k\}$,  resulting in high computation complexity.  In contrast,  by exploiting  the threshold-based structure of the optimal resource-allocation policy in Theorem~\ref{Theo:OptiPolicyP2}, the proposed solution approach, described in Algorithm~\ref{Alg:SLA:Opt}, needs to perform only a  \emph{one-dimension} search for  $\lambda^{*}$, reducing the computation complexity significantly. To facilitate the search, next lemma gives the range of $\lambda^{*}$, which can be easily proved from Theorem~\ref{Theo:OptiPolicyP2} and omitted for simplicity.}
\end{remark}

\begin{lemma}\label{Lem:ProLambda}
\emph{When there is at least one offloading mobile, the optimal Lagrange multiplier $\lambda^*$ satisfies:
$$0\le \lambda^* \le \lambda_{\max}=\max_k  \phi_k.$$
}
\end{lemma}
\vspace{-10pt}
\begin{algorithm}
  \caption{Optimal Algorithm for Problem P2.}
  \label{Alg:SLA:Opt}
  \begin{itemize}
\item{\textbf{Step 1} [Initialize]: \\
Let $\lambda_\ell=0$ and $\lambda_h=\lambda_{\text{max}}$.
 According to Theorem~\ref{Theo:OptiPolicyP2}, obtain $T_{\ell}=\sum_{k=1}^{K} t_{k,\ell}^{*(2)}$  and  $T_{h}=\sum_{k=1}^{K} t_{k,h}^{*(2)}$, where $\{t_{k,\ell}^{*(2)} \}$ and $\{t_{k,h}^{*(2)}\}$ are the allocated fractions of slot for the cases of $\lambda_\ell$ and $\lambda_h$, respectively. }
\item{\textbf{Step 2} [Bisection search]:\\  
While $T_{\ell} \neq T$ and  $T_{h} \neq T$, update $\{\lambda_\ell, \lambda_h\}$ as follows. (1) Define $\lambda_{m}=(\lambda_\ell+\lambda_h)/2$ and compute $T_{m}$. \\(2) If $T_m = T$, then $\lambda^*=\lambda_m$ and the optimal policy can be determined. Otherwise, if $T_{m}<T$, let $\lambda_{h}=\lambda_{m}$ and if $T_{m}>T$, let $\lambda_{\ell}=\lambda_{m}$.
}
\end{itemize}
  \end{algorithm}
\vspace{-3pt}

\subsection{Special Cases}
The optimal resource-allocation policies for several special cases considering equal weight factors are discussed as follows. 
\subsubsection{Uniform  channels and local computing} Consider the simplest case where $\{ h_k, C_k, P_k\}$ are identical for all $k$. Then all mobiles have uniform offloading priorities. In this case, for optimal resource allocation, different mobiles can offload arbitrary data sizes so long as the sum offloaded data size satisfies the following constraint: 
\begin{equation}
\sum_{k=1}^K \ell_k^{*(2)} \leq T B\log_2\l(\frac{  B C P h^2}{N_0 \ln 2}\r). \nn
\end{equation}

\subsubsection{Uniform channels} Consider the case of   $h_1 = h_2\cdots = h_K$. The offloading priority for each mobile, say mobile $k$,   is only affected by the corresponding  local-computing parameters $P_k$ and $C_k$. Without loss of generality, assume that $P_1 C_1 \leq P_2 C_2 \cdots \leq P_K C_K$. Then the optimal resource-allocation policy is given in the following corollary of Theorem~\ref{Theo:OptiPolicyP2}. 

\begin{corollary}\label{Coro: SpeCase2}\emph{Assume infinite cloud computation capacity,  $h_1 = h_2\cdots = h_K$ and $P_1 C_1 \leq P_2 C_2 \cdots \leq P_K C_K$. Let $k_t$ denote the index such that $\phi_k < \lambda^\ast$ for all $k < k_t $ and $\phi_k >  \lambda^\ast$ for all $k \geq  k_t $. The optimal resource-allocation policy is given as follows: 
\begin{align*}
\ell_{k}^{*(2)}&=
   \begin{cases}
   R_k,  & k \geq k_t\\
   m_k^+ , &\mbox{otherwise}
   \end{cases},
\end{align*}
and $$ t_k^{*(2)}=\frac{ \ln 2}{B\l[W_0\l(\frac{\lambda^{*} h^2/\beta-N_0}{N_0 e}\r)+1\r]}\times \ell_{k}^{*(2)}.$$
}
\end{corollary}
The result shows that the  optimal resource-allocation policy follows  a \emph{greedy} approach that selects  mobiles in a descending order of energy consumption per bit for complete offloading until the time-sharing duration is fully utilized.

\subsubsection{Uniform local computing} Consider the case of   $C_1P_1 = C_2P_2\cdots = C_K P_K$. Similar to the previous case, the optimal resource-allocation policy can be shown to follow the greedy approach that selects mobiles for complete offloading in the descending order of channel gain. 

\section{Multiuser  MECO: Finite  Cloud   Capacity}
In this section, we consider the case of finite cloud computation capacity and analyze the optimal resource-allocation policy for solving Problem P1. The policy is shown to also have a threshold-based structure as the infinite-capacity counterpart derived in the preceding section.  Both the optimal and sub-optimal algorithms are presented for policy computation. 

\subsection{Optimal Resource-Allocation Policy}
To solve  the convex Problem P1, the  corresponding Lagrange function can be written as 
\begin{equation}\nn
\begin{aligned}
\tilde{L}=\sum_{k=1}^{K} \beta_k &\l[ \dfrac{t_k}{h_k^2}  f\l(\frac{\ell_k}{t_k}\r) + (R_k-\ell_k) C_k P_k \r]\\ &+ \lambda\l(\sum_{k=1}^K t_k - T\r) +  \mu \l(\sum_{k=1}^{K} C_k \ell_k - F\r).
\end{aligned}
\end{equation}
where $\mu \geq 0 $  is the Lagrange multiplier corresponding to the cloud computation capacity constraint. Using the above Lagrange function,   it is straightforward to show that the corresponding KKT conditions can be modified from their infinite-capacity counterparts in \eqref{Eq:OptL}-\eqref{Eq:OptT:a} by  replacing $P_k$ with $\tilde{P}_k=P_k-\mu$, called the \emph{effective computation energy per cycle}. The resultant \emph{effective offloading priority function}, denoted as $\tilde{\phi}_{k}$,  can be modified accordingly from that in \eqref{Eq:OffPriority} as 
\begin{equation}\label{Eq:EffectiveOffPriority}
\!\!\!\!\tilde{\phi}(\beta_k, C_k, P_k, h_k,\mu)\!\!=\!\!\begin{cases}
   \dfrac{\beta_k N_0}{h_k^2} \!\l( \tilde{\upsilon}_k  \ln \tilde{\upsilon}_k\!-\!\tilde{\upsilon}_k\!+\!1\r),  \!\!\!\!\!&\mbox{$\tilde{v}_k\ge 1$}\\
   0,  &\mbox{$\tilde{v}_k<1$}
   \end{cases},
\end{equation} 
where $\tilde{\upsilon}_k=\dfrac{B C_k  (P_k-\mu) h_k^2 }{N_0\ln2 }$. Based on above discussion, the main result of this section follows as shown below. 
\begin{theorem}[Optimal Resource-Allocation Policy]\label{Theo:OptiPolicyP1}\emph{Consider the case of finite cloud computation capacity.  The optimal policy solving Problem P1 has the same structure as that in Theorem~\ref{Theo:OptiPolicyP2} and is expressed in terms of  the priority function $\tilde{\phi}_{k}$ in 
\eqref{Eq:EffectiveOffPriority} and the optimized Lagrange multipliers $\{\lambda^*, \mu^*\}$. }
\end{theorem}

Computing the threshold for the optimal resource-allocation policy requires  a \emph{two-dimension search} over the Lagrange multipliers  $\{\lambda^{*}, \mu^{*}\}$, using  Algorithm~\ref{Alg:WLA:Opt}. For an efficient search, it is useful to limit the range of $\lambda^*$ and $\mu^*$ as follows. 

\begin{lemma}\label{Lem:ProMu}
\emph{When there is at least one offloading mobile, the optimal Lagrange multipliers $\{\lambda^*, \mu^*\}$ satisfy:
\begin{align*}
& 0\le \lambda^* \le \lambda_{\max},\\
&0\le \mu^* \le \mu_{\max}=\max_k \l\{P_k-\dfrac{N_0 \ln 2}{B C_k h_k^2}\r\}
\end{align*}
where $\lambda_{\max}$ has been defined in Lemma~\ref{Lem:ProLambda}.
}
\end{lemma}
\begin{proof}
See Appendix~\ref{App:Lem:ProMu}
\end{proof}

Note that $\mu^* = 0$ corresponds to the case of infinite cloud computation capacity and $\mu^* = \mu_{\max}$ to the case where offloading yields no energy savings for any mobile. 

\begin{algorithm}
  \caption{ Optimal Algorithm for Solving Problem P1.}
  \label{Alg:WLA:Opt}
  \begin{itemize}
\item {\textbf{Step 1}[Check solution for Problem P2]: \\Perform Algorithm~\ref{Alg:SLA:Opt}. If $\sum_{k=1}^K \ell_k^{*(2)} \le F$, the optimal policy is given in Theorem~\ref{Theo:OptiPolicyP2}. Otherwise, go to Step 2.
}
\item{\textbf{Step 2} [Initialize]:\\ Let $\mu_\ell=0$ and $\mu_h=\mu_{\max}$.  Based on Theorem~\ref{Theo:OptiPolicyP1}, obtain $F_\ell=\sum_{k=1}^K C_k \ell_{k,\ell}^*$ and $F_h=\sum_{k=1}^K C_k \ell_{k,h}^*$, where $\{\ell_{k,\ell}^*\}$ and $\{\ell_{k,h}^*\}$ are the offloaded data sizes for $\mu_\ell$ and $\mu_h$, respectively, involving the one-dimension search for $\lambda^{*}$. }
\item{\textbf{Step 3} [Bisection search]:\\ While $F_{\ell} \neq F$ and  $F_{h} \neq F$, update $\{\mu_\ell, \mu_h\}$ as follows. (1) Define $\mu_{m}=(\mu_\ell+\mu_h)/2$ and compute $F_{m}$. \\(2) If $F_m = F$, then $\mu^*=\mu_m$  and the optimal policy can be determined. Otherwise, if $F_{m}<F$, let $\mu_{h}=\mu_{m}$ and if $F_{m}>F$, let $\mu_{\ell}=\mu_{m}$.}
\end{itemize}
  \end{algorithm} 
\vspace{-10pt}
\subsection{Sub-Optimal Resource-Allocation Policy}
To reduce the computation  complexity of Algorithm \ref{Alg:WLA:Opt} due to the two-dimension search,  one simple sub-optimal  policy is designed using Algorithm~\ref{Alg:WLA:Sopt}. The key idea is to decouple the computation and radio resource allocation.  In Step~$2$, based on the \emph{approximated} offloading priority in \eqref{Eq:OffPriority} for the case of infinite cloud computation capacity, we allocate the computation resource to mobiles with high offloading priorities. Step $3$  optimizes  the  corresponding fractions of slot given offloaded data. This sub-optimal algorithm has  low complexity requiring only a one-dimension search.  Moreover,  its performance is shown by simulation  to be close-to-optimal in the sequel.

\begin{algorithm}[h]
  \caption{ Sub-optimal Algorithm for Solving  Problem P1.}
  \label{Alg:WLA:Sopt}
  \begin{itemize}
  \item{\textbf{Step 1}:  Perform Algorithm~\ref{Alg:SLA:Opt}. If $\sum_{k=1}^K \ell_k^{*(2)} \!\!\!\le\!\!\! F$, Theorem~\ref{Theo:OptiPolicyP2} gives the optimal policy. Otherwise, go to Step 2.}
\item{\textbf{Step 2}: Based on  offloading priorities in \eqref{Eq:OffPriority}, offload the data from mobiles in the descending order of offloading  priority until  the cloud computation capacity is fully occupied, i.e., $\sum_{k=1}^{K} C_k \ell_k^* \!=\! F.$}
\item{\textbf{Step 3}: With $\{\ell_k^{*}\}$ derived in Step 2, search for $\lambda^{*}$ such that  $t_k^{*}\!\!=\!\! \dfrac{\ell_k^{*}\ln 2}{B[W_0(\frac{\lambda^{*} h_k^2/\beta_k-N_0}{N_0 e})+1]}$ satisfying $\sum_{k=1}^K t_k^{*}\!\!=\!\!T$.
}
\end{itemize}
  \end{algorithm}

\section{Simulation Results}

The simulation settings are  as follows unless specified  otherwise.  The  MECO system comprises $K=30$ mobiles with equal fairness weight  factors, namely that $\beta_k=1$ for all $ k$ such that the weighted sum mobile energy consumption represents the total mobile energy consumption. The time slot  $T\!=\!100$ ms and channels are modeled as   independent Rayleigh fading   with average power loss set as $10^{-6}$. In addition, the variance of complex white Gaussian channel noise is $N_0\!\!=\!\!10^{-9}$ W and the bandwidth $B\! =\!10$ Mhz. Consider mobile $k$. The  CPU computation capacity $F_k$ is  uniformly selected  from the set $\{0.1, 0.2, \cdots, 1.0\}$ Ghz  and the local computing energy per cycle $P_k$ follows a uniform distribution in the range $(0, 20\times 10^{-11})$  J/cycle. For the computing task, both the data size and required number of  CPU cycles per bit follow the uniform distribution with $R_k \!\in\! [100, 500]$ KB and $C_k \!\in\! [500, 1500]$ cycles/bit. All random variables are  independent for different mobiles, modeling heterogeneous mobile computing capabilities.  Last, the cloud computation capacity is set as $F=6\times 10^9$ cycles per slot.

For performance comparison, a  baseline \emph{equal resource-allocation} policy is considered, which  allocates equal offloading duration for mobiles satisfying $\upsilon_k > 1$ and based on this, the offloaded data sizes are optimized.

Fig.~\ref{Fig:egy_vs_T} shows the curves of total mobile energy consumption versus the time slot  duration $T$. Several observations can be made. First, the total mobile energy consumption reduces  as the slot duration grows. Next,  the sub-optimal policy computed using Algorithm~\ref{Alg:WLA:Sopt} is found to  have close-to-optimal performance and yields total mobile energy consumption  less than half of that for  the equal resource-allocation policy. The energy reduction is more significant for a shorter slot duration since without the optimization on fractions of slot, the offloading energy of baseline policy grows exponentially with the decrease of allocated time fractions.

 The curves of total mobile energy consumption versus the cloud computation capacity are displayed in Fig.~\ref{Fig:egy_vs_capacity}. It can be observed that the performance of the sub-optimal policy  approaches to that of the optimal one when the cloud computation capacity increases and achieves  substantial energy savings gains over the equal resource-allocation policy. Furthermore, the total mobile energy consumption is invariant after the cloud computation capacity exceeds some  threshold (about $6\times 10^9$). This suggests that   there exists some critical value for the cloud computation capacity,  above which   increasing the capacity yields no reduction on the total mobile energy consumption. 
 \begin{figure}[t!]
\begin{center}
\includegraphics[width=8cm]{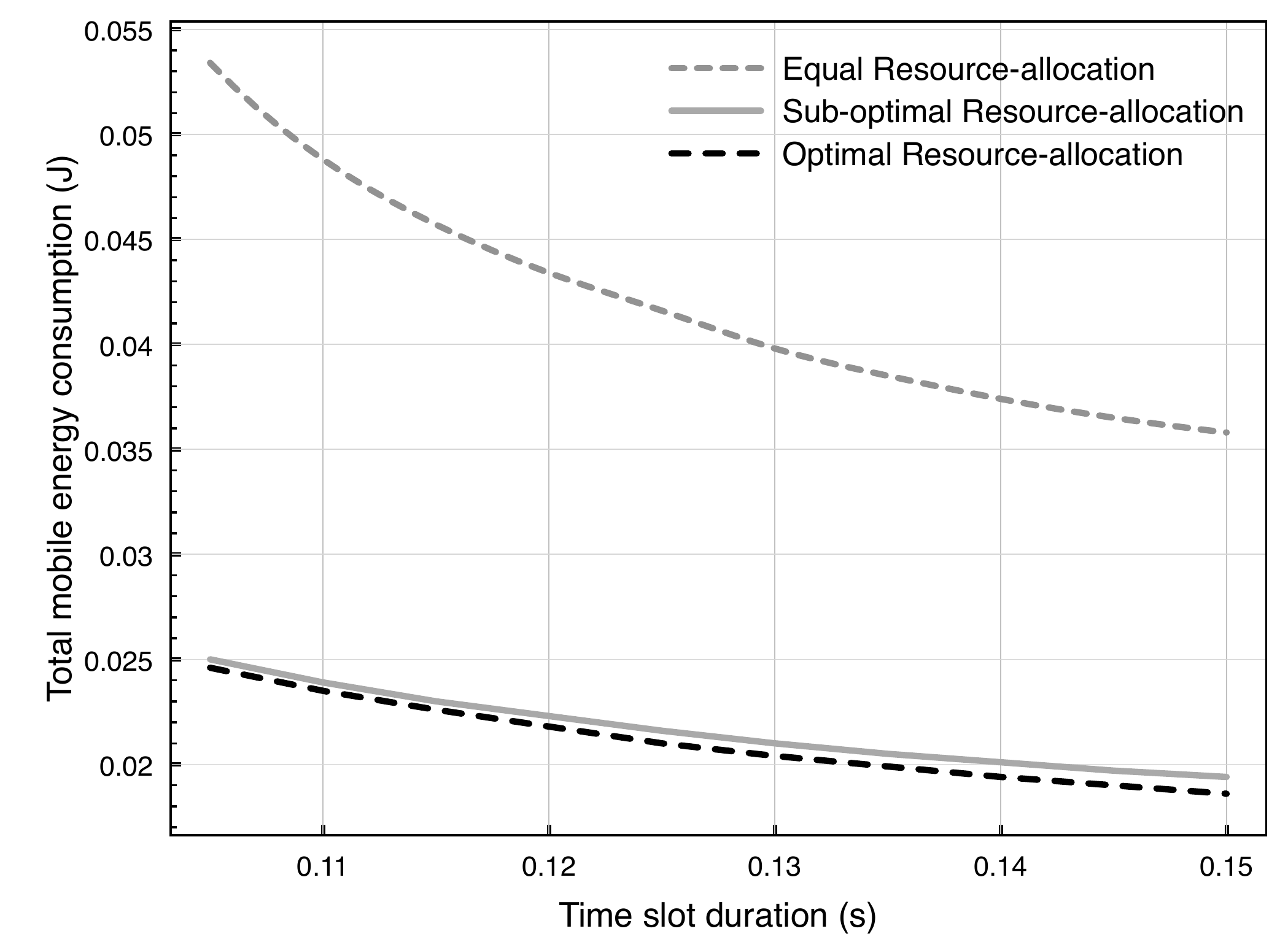}
\caption{Total mobile energy consumption vs. time slot duration.}
\label{Fig:egy_vs_T}
\end{center}
\end{figure}
 \begin{figure}[t!]
\begin{center}
\includegraphics[width=8cm]{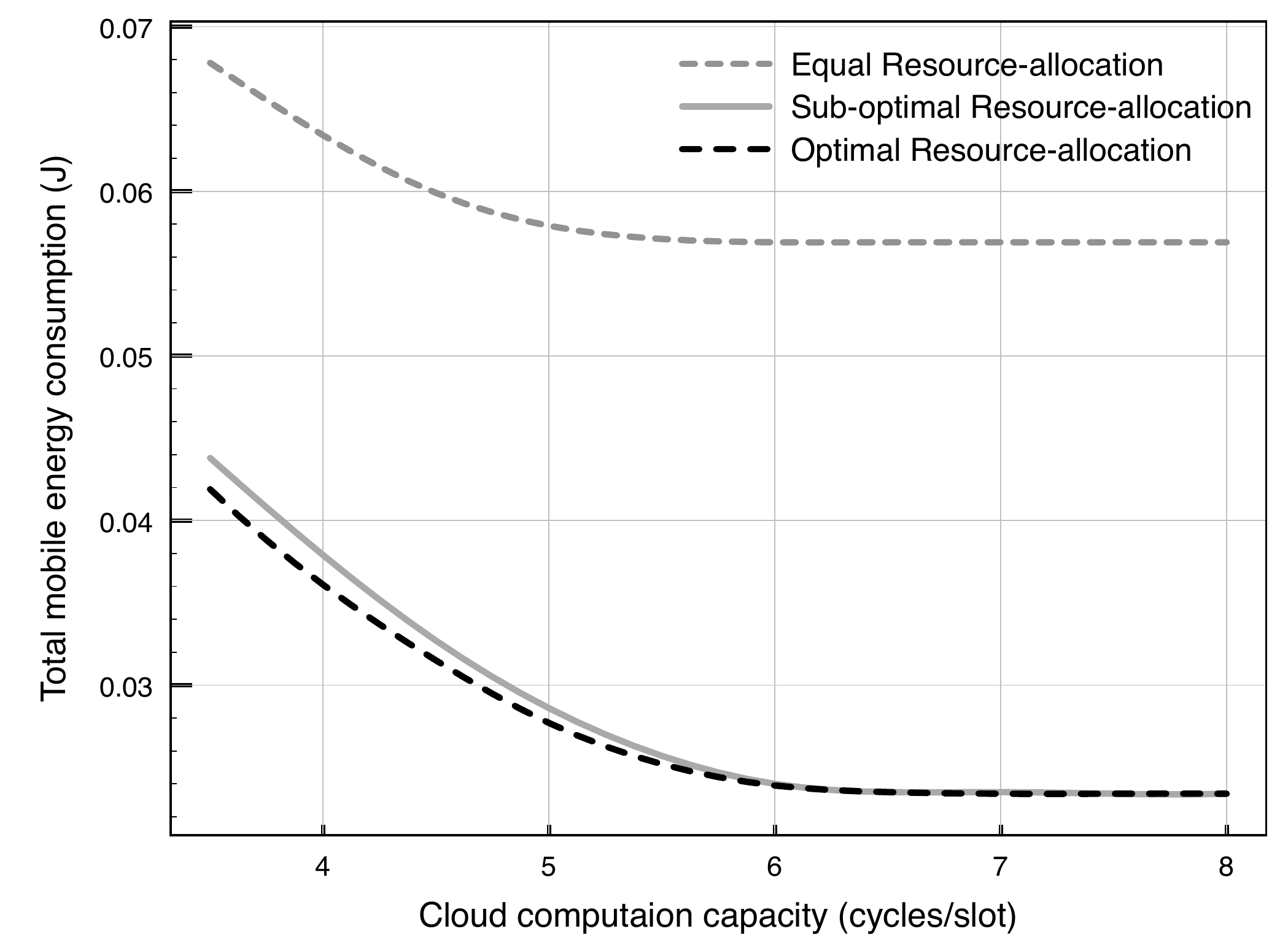}
\caption{Total mobile energy consumption vs. cloud computation capacity.}
\label{Fig:egy_vs_capacity}
\end{center}
\end{figure}

%

\section{conclusion}
Consider a  multiuser MECO system based on TDMA. This work shows that the optimal energy-efficient resource-allocation policy for clouds with infinite or finite computation capacities, is featured with a threshold-based structure. Specifically, the BS makes a binary offloading decision for each mobile, where users with priorities above or below a given threshold will perform complete or minimum offloading. Moreover,  a simple sub-optimal algorithm is proposed to reduce the complexity for computing the threshold.
\appendix

\subsection{Proof of Lemma~\ref{Lem:ConvexProblem}}\label{App:ConvexProblem}
Since $f(x)$ is a convex function, its perspective function, defined as $t_k f(\frac{\ell_k}{t_k})$, is still convex. 
Thus, the objective function, the summation of a set of convex functions, preserves the convexity. Combining it with the linear convex constraints  leads to the desired result. \hfill $\blacksquare$
%

\subsection{Proof of Lemma~\ref{Lem:FeasiTypeI}}\label{App:FeasiTypeI}
Whether Problem P1 is feasible depends on the following two key constraints: $\sum_{k=1}^K C_k \ell_k \le F$ and $m_k^+ \le \ell_k \le R_k$. Assume $m_k^+ \le \ell_k \le R_k$ is satisfied. Then it has 
\begin{equation*}
\sum_{k=1}^K C_k m_k^+ \le \sum_{k=1}^K C_k \ell_k \le \sum_{k=1}^K C_k R_k.
\end{equation*}
Thus, only when $\sum_{k=1}^K C_k m_k^+ \le F$, Problem P1 is feasible.\hfill $\blacksquare$

\subsection{Proof of Lemma~\ref{Lem:PriStrongLA}}\label{App:PriStrongLA}
First, we derive a general result that is the root of equation: $f^{'-1}\!\l(p \r)=g^{-1}\!\l( y \r)$ with respect to $y$ as follows. 

According to the definitions of $f(x)$ and $g(x)$, it has 
\begin{equation}\label{Eq:Ffunction}
f^{'}(x)=\dfrac{N_0 \ln 2}{B} 2^{\frac{x}{B}} ~~\text{and}~~ f^{'-1}(y)=B\log_2\l(\dfrac{By}{N_0 \ln 2}\r).
\end{equation} Thus, the solution for the general equation is
\begin{align}\label{Eq:general}
y&=g(f^{'-1}(p))=f(f^{'\!\!-1}(p))-f^{'\!-\!1}(p) \times f^{'}\!\!(f^{'\!-\!1}(p)) \nn \\ 
&=f(f^{'-1}(p))- f^{'-1}(p) \times p \nn\\
&=\dfrac{Bp}{\ln 2}-N_0-pB\log_2\l(\dfrac{Bp}{N_0 \ln 2}\r).
\end{align}

Note that to ensure $\ell_k^{*(2)}\ge 0$ in Problem P1, we need $f^{'-1}(C_kP_k h_k^2)\!\ge\! 0$, which is equivalent to $v_k \!\ge\! 1$ derived from \eqref{Eq:Ffunction}. Then, by substituting $p=C_k P_k h_k^2$ and $y=\frac{-h_k^2 x}{\beta_k}$ to \eqref{Eq:general} and making arithmetic operations  gives the desired result.
\hfill $\blacksquare$


\subsection{Proof of Theorem~\ref{Theo:OptiPolicyP2}}\label{App:OptiPolicyP2}
First, to prove this theorem, we need the following lemmas.
\begin{lemma}\label{Lem:GInverse}\emph{The function $g^{-1}(y)$ can be expressed as 
\begin{equation}\label{Eq:Ginverse}
g^{-1}(y)=\frac{B \l[W_0(\frac{y+N_0}{-N_0 e})+1\r]}{\ln 2}.
\end{equation}}
\end{lemma}
\begin{proof}
Since $g^{-1}(y)$ denotes the root of equation $g(x)=y$ for $x \ge 0$, it has 
\vspace{-5pt}
\begin{align*}
y&=g(x)=\l[N_0- \frac{(\ln 2) N_0 x}{B}\r] \times 2^{\frac{x}{B}}-N_0 \\
&= (-N_0 e) \times \l[ \l(\frac{x\ln2 }{B}-1\r) e^{\frac{x \ln2}{B}-1}\r] -N_0.
\end{align*}
Thus, based on the definition for Lambert function, we have $\dfrac{x\ln2 }{B} \!-1\!=\!W_0\!\!\l(\dfrac{y+N_0}{-N_0 e}\r)$. Then the desired result follows.
\end{proof}

\begin{lemma}\label{Lem:IncreaseG}\emph{
The function $g^{-1}(y)$ is a monotone decreasing function for $y<0$.}
\end{lemma}
\begin{proof}
From \eqref{Eq:Ginverse}, for $y\le 0$, it has $\frac{y+N_0}{-N_0 e} \ge -1/e$. Since the single-valued Lambert function $W_0(x)$ is monotone increasing for $x\ge -1/e$, we can easily obtain the desired result.
\end{proof}

Then, consider case 1) in Theorem~\ref{Theo:OptiPolicyP2}. Note that  for mobile $k$, if  $m_k^+=0$ and $\upsilon_k\le 1$, it results in $\ell_k^{*(2)}=0$ derived from \eqref{Eq:OptL}. Thus, if these two conditions are satisfied for all $k$, it leads to $\ell_k^{*(2)}=t_k^{*(2)}=0$. For case 2), if there exists  mobile $k$ such that $\upsilon_k >1$ or $m_k^+>0$, it ensures $\ell_k^{*(2)}>0$. And the time-sharing constraint should be active since remaining time can be used for offloading so as to reduce the transmission energy. 
Moreover, consider each user $k=1, 2, \cdots K$. If $\upsilon_k \ge 1$, from \eqref{Eq:OptL} and \eqref{Eq:OptT}, $\{\ell_k^{*(2)}, t_k^{*(2)}\}$ should satisfy the following:
\begin{subequations}
\begin{align}
\dfrac{\ell_k^{*(2)}}{ t_k^{*(2)}}&=\min\l\{ \max \l[\dfrac{m_k^+}{t_k^{*(2)}}, f^{'-1}\!\l(C_k P_k h_k^2 \r)\r], \dfrac{R_k}{t_k^{*(2)}}\r\} \label{Eq:TL1} \\
&=\max \l\{ \dfrac{m_k^+}{t_k^{*(2)}}, \min \l[f^{'-1}\!\l(C_k P_k h_k^2 \r), \dfrac{R_k}{t_k^{*(2)}}\r]\r\} \label{Eq:TL2} \\
&=g^{-1}\!\l(\frac{-h_k^2 \lambda^*}{\beta_k} \r). \label{Eq:TL3}
\end{align}
\end{subequations}
Using Lemma~\ref{Lem:PriStrongLA} and Lemma~\ref{Lem:IncreaseG}, we have the following:
\begin{enumerate}
\item{If $\phi_k>\lambda^*\ge 0$, it has $-h_k^2 \phi_k<-h_k^2 \lambda^*\le 0$.  Then,  from \eqref{Eq:TL1}, it gives 
 \begin{align}\label{Eq:HighPri}
 \max & \l[\dfrac{m_k^+}{t_k^{*(2)}}, f^{'-1}\!\l(C_k P_k h_k^2 \r)\r] \ge f^{'-1}\!\l(C_k P_k h_k^2\r) \nn \\
 &=g^{-1}\!\l(\frac{-h_k^2 \phi_k}{\beta_k} \r)>g^{-1}\!\l(\frac{-h_k^2 \lambda^*}{\beta_k} \r).
 \end{align} From \eqref{Eq:TL1}, \eqref{Eq:TL3} and \eqref{Eq:HighPri}, it follows that
$\ell_k^{*(2)}=R_k$.}
\item{If $\phi_k=\lambda^*$, it has $f^{'-1}\!\l(C_k P_k h_k^2\r)=g^{-1}\!\l(\frac{-h_k^2 \lambda^*}{\beta_k} \r)$.} 
\item{If $0\le \phi_k<\lambda^*$, it has $-h_k^2 \phi_k>-h_k^2 \lambda^*$. Combining it with \eqref{Eq:TL2} leads to 
 \begin{align}\label{Eq:LowPri}
 \min& \l[ f^{'-1}\!\l(C_k P_k h_k^2 \r),\dfrac{R_k}{t_k^{*(2)}},\r] \le f^{'-1}\!\l(C_k P_k h_k^2\r) \nn \\
 &=g^{-1}\!\l(\frac{-h_k^2 \phi_k}{\beta_k} \r)<g^{-1}\!\l(\frac{-h_k^2 \lambda^*}{\beta_k} \r).
 \end{align} From \eqref{Eq:TL2}, \eqref{Eq:TL3} and \eqref{Eq:LowPri}, it follows that
$\ell_k^{*(2)}=m_k^+$.}
\end{enumerate}
Furthermore, if $\upsilon_k\!<\!1$, it has $\ell_k^{*(2)}\!=\!m_k^+$. Note that this case can be included in the scenario of $\phi_k\!<\!\lambda^*$ with the definition of $\phi_k$ in \eqref{Eq:OffPriority}.

Last, from \eqref{Eq:TL3}, it follows that $$t_k^{*(2)}=\dfrac{\ell_k^{*(2)}}{g^{-1}\!\l(\frac{-h_k^2 \lambda^*}{\beta_k} \r)}\overset{(a)}{=}\dfrac{\ell_k^{*(2)} \ln 2}{B\l[W_0(\frac{\lambda^* h_k^2/\beta_k-N_0}{N_0 e})+1\r]}$$ where $(a)$ is derived using Lemma~\ref{Lem:GInverse}, completing the proof. \hfill $\blacksquare$

\subsection{Proof of Lemma~\ref{Lem:ProMu}}\label{App:Lem:ProMu}
If there exists offloading mobile $k$, it must satisfy $\lambda^{*} \le \tilde{\phi}_k$ and $1 \le \tilde{\upsilon}_k$. Thus, considering all mobiles, it follows $\lambda^{*} \le \max_k\{\tilde{\phi}_k\}=\lambda_{\max}$ and $1 \le \max_k \{\frac{B C_k  (P_k-\mu^*) h_k^2 }{N_0\ln2 }\}.$ The latter condition is equivalent to $\mu^*\le \mu_{\max}$, completing the proof.\hfill $\blacksquare$


\end{document}